\documentclass[11pt,a4paper,english]{amsart}
\usepackage{babel}
\usepackage[latin1]{inputenc}
\usepackage{amsmath}
\usepackage{amsthm}
\usepackage{amsfonts}
\usepackage{indentfirst}
\usepackage{graphicx}
\usepackage{amssymb}
\usepackage[mathscr]{eucal}

\usepackage{tikz}

\usepackage{enumerate}
\usepackage{amsthm}
\usepackage{amscd}
\newtheorem{teo}{Theorem}[section]
\newtheorem{pro}[teo]{Proposition}%[section]
\newtheorem{cor}[teo]{Corollary}%[section]
\newtheorem{lem}[teo]{Lemma}%[section]
\newtheorem{exa}[teo]{Example}%[section]
\theoremstyle{definition}
\newtheorem{de}[teo]{Definition}%[section]

\title[Locally recoverable $J$-affine variety codes]{Locally recoverable $J$-affine variety codes}
\author{Carlos Galindo, Fernando Hernando and Carlos Munuera}
\curraddr{\texttt{Carlos Galindo and Fernando Hernando:} Instituto
Universitario de Matem\'aticas y Aplicaciones de Castell\'on and
Departamento de Matem\'aticas, Universitat Jaume I, Campus de Riu
Sec. 12071 Castell\'{o} (Spain)\\
\texttt{Carlos Munuera:} Instituto de Matem\'aticas (Imuva) and Departamento de Matem\'atica Aplicada, Universidad de Valladolid, Avda Salamanca SN, 47014 Valladolid  (Spain).
}
\email{{\rm Galindo:} galindo@uji.es; {\rm Hernando:} carrillf@uji.es; {\rm Munuera:} cmunuera@arq.uva.es}
\date{}
\thanks{
Supported by the Spanish Goverment, Ministerio de Ciencia, Innovaci\'on y Universidades (MICINN)/FEDER (grants PGC2018-096446-B-C21 and PGC2018-096446-B-C22), Generalitat Valenciana (grant AICO-2019-223) and  University Jaume I (grant PB1-1B2018-10)}


\begin{document}

\begin{abstract}
A  locally recoverable (LRC) code is a code over a finite field $\mathbb{F}_q$ such that any erased coordinate of a codeword can be recovered from a small number of other coordinates in that codeword.
We construct LRC codes correcting more than one erasure, which are subfield-subcodes  of some $J$-affine variety
codes.  For these LRC codes, we compute localities $(r, \delta)$ that determine the minimum size of a set $\overline{R}$ of positions so that any $\delta- 1$ erasures in  $\overline{R}$ can be recovered from the remaining $r$  coordinates in this set.
We also show that some of these LRC codes with lengths $n\gg q$ are $(\delta-1)$-optimal.
\end{abstract}

\maketitle

\section*{Introduction}

The growth of the amount of stored information in large scale distributed and cloud storage systems makes the loss of data due to node failures a major problem. To obtain a reliable storage, when a node fails, we want to recover the data it contains
by using information from other nodes. This is the {\em repair problem}. A naive solving method consists of the replication of information across several nodes. A more clever method is to protect the data using error-correcting codes, what has led to the introduction of  locally recoverable (LRC) codes \cite{GHSY}. LRC codes are error-correcting codes for which one or more erased coordinates of a codeword can be recovered from a set of other coordinates in that codeword. As typical examples of this solution we can mention Google and Facebook storage systems that use Reed-Solomon (RS) codes to protect the information. The procedure is as follows: the information to be stored is a long sequence $b$  of elements belonging to a finite field $\mathbb{F}_{p^l}$, where $p$ is a prime number. This sequence is divided into blocks, $b = b_1, b_2, \ldots, b_z$, of the same length $h$. According to the isomorphism $\mathbb{F}_{p^l}^h \cong \mathbb{F}_{p^{lh}}$,  each of these blocks can be seen as an element of the finite field $\mathbb{F}_q$, with $q= p^{s}$ and $s=lh$. Fix an integer $k < q$. The vector $(b_1, b_2, \ldots, b_k) \in \mathbb{F}_q^k$ is encoded by using a Reed-Solomon code of dimension $k$ over $\mathbb{F}_q$, whose length $n$, $k < n \le q$, is equal to the number of nodes that will be used in its storage. We choose $\alpha_1, \alpha_2, \ldots, \alpha_n \in \mathbb{F}_q$  and send
$$
f(\alpha_i)=b_1 + b_2\alpha_i + \dots + b_k \alpha_i^{k-1}
$$
to the $i$-th node.  Even if a node fails, we may recover   the stored data   $(b_1, b_2, \ldots, b_k)$ by using Lagrangian interpolation from any other $k$ available nodes.

Note that this method is wasteful, since $k$ symbols  over $n$ nodes must be used to recover just one  erasure.
Of course other error-correcting codes, apart from RS codes, can be used to deal more efficiently with the repair problem. Thus, in terms of coding theory the repair problem can be stated as follows: let $\mathcal{C}$ be a linear code of length $n$ and dimension $k$ over $\mathbb{F}_q$. A coordinate $i \in \{1, 2, \ldots, n\}$ is locally recoverable  if there is a recovery set $R=R(i) \subset \{1, 2, \ldots, n\}$  such that  $i \notin R$ and for any codeword $\boldsymbol{x} \in \mathcal{C}$, an erasure at position $i$  of $\boldsymbol{x}$ can be recovered by using the information given by the coordinates of $\boldsymbol{x}$ with indices in $R$.  The  locality of the coordinate $i$ is the smallest size of a recovery set for $i$. The code $\mathcal{C}$ is  {\em locally recoverable} (LRC) if each coordinate is so,  and the {\it locality} of $\mathcal{C}$ is the maximum locality of its coordinates. In Section \ref{secuno} we shall specify these definitions. Note that strictly speaking, all codes $\mathcal{C}$ of minimum distance $d (\mathcal{C})> 1$ are locally recoverable; just take $\{1, 2, \ldots,n\} \setminus\{ i\}$ as a recovery set for coordinate $i$. However we are interested in codes admitting recovery sets as small as possible. Thus,  in practice we restrict to consider codes with `moderate' localities. In general, the locality $r$ of an LRC code $\mathcal{C}$ with parameters $[n,k,d]$ is upper-bounded as $r\le k$. For example, MDS codes (RS codes in particular) of dimension $k$ have locality $k$.  Several  lower bounds on $r$ are known. The most commonly used is the Singleton-like bound (\ref{Singd1Eq}).

Among the different classes of codes considered as good candidates for local recovering, cyclic codes and subfield-subcodes of cyclic codes play an important role, because the cyclic shifts of a recovery set again provide recovery sets \cite{CXHF,GC,HYUS,TBGC}.
In this article we continue this line of research by using the very general language of affine variety codes. We  consider specific $J$-affine variety codes, introduced in \cite{QINP}, whose subfield-subcodes provide LRC codes. These subfield-subcodes have large lengths  over fields $\mathbb{F}_q$,  and Theorems \ref{el21} and \ref{el25} provide bounds on their localities.

A variant of LRC codes was introduced in \cite{PGLK}. As multiple device failures may occur simultaneously, it is of interest to consider LRC codes correcting more than one erasure. This idea leads to the concept of  localities $(r,\delta)$ of an LRC code $\mathcal{C}$, which measure the  recovery capability of $\mathcal{C}$ when at most $\delta-2$ erasures occur in a recovery set (see Section \ref{secuno} for a rigorous definition). LRC codes for multiple erasures have been subsequently studied in \cite{CXHF, FF, ABHM}. In \cite{CXHF} the authors constructed some classes of such LRC codes over $\mathbb{F}_q$, with lengths $n$ such that either $n| q-1$ or $n| q+1$. Codes of  similar type and unbounded length were given in \cite{FF}. Here $\delta=d-1,d-2$ or $\delta=d/2$, where $d$ stands for the minimum distance.

The localities $(r,\delta)$ of an LRC code satisfy a Singleton-like bound ((\ref{SingdeltaEq}) in Section \ref{secuno}). Codes reaching equality for some $(r,\delta)$, are called {\em optimal}. For example, the codes in \cite{CXHF, FF} are optimal. Note that, as for the original Singleton bound, the bounds (\ref{Singd1Eq}) and (\ref{SingdeltaEq}) do not depend on the cardinality of the ground field $\mathbb{F}_q$. Some  size dependent bounds can be found in \cite{ABHM}.

A somewhat different definition of LRC code with localities $(r,\delta)$
is proposed in \cite{KPLK} for systematic codes. There, the purpose is to repair erasures on the information symbols of a codeword. Other related variants of LRC codes deal with sequential repair of erasures \cite{PLK}, the availability property \cite{WZ}, or the cooperative repair \cite{RMV}.

In this work we use use affine variety constructions to obtain LRC codes suitable for multiple erasures, whose localities $(r,\delta)$ behave well (Theorems \ref{teo29} and \ref{la211}). In some cases these codes are optimal
for the Singleton-like bound (\ref{SingdeltaEq}). Compared with the codes shown in \cite{CXHF}, the ours are considerably longer, although not optimal in general. Let us recall here that most
good currently known LRC codes have small lengths $n$, in comparison with the cardinality
of the ground field $q$; usually $n < q$, \cite{GXY} (or $n=q+1$ for some codes in \cite{CXHF}). For the opposite, our codes (as is the case with those in \cite{FF}) have lengths $n\gg q$.

The article is organized as follows: in Section \ref{secuno} we recall some basic facts about LRC codes and introduce the concept of $t$-locality. Section \ref{secdos} is devoted to develop and study LRC codes from affine varieties. In Subsection \ref{afine} we introduce $J$-affine variety codes which also gave rise to good quantum error-correcting codes in \cite{galindo-hernando, QINP, QINP2}. In subsections \ref{subsect2} and \ref{subsect3}, we show that subfield-subcodes of several types of $J$-affine variety codes are good LRC codes, and we determine some of their localities $(r,\delta)$.
Finally in Section \ref{sectres} we give examples of  LRC codes obtained by our procedure. We list some parameters and localities.

\section{LRC codes}
\label{secuno}

In this section we state some definitions and facts concerning LRC codes that will be necessary for the rest of the work.  We mostly follow the usual conventions and definitions of locally recoverable codes. As a notation, given a  fixed coordinate $i$ and  a  set $R$ such that $i\notin R$, we write $\overline{R}=R\cup \{ i\}$.   Let $\mathcal{C}$ be an $[n, k, d]$ code over $\mathbb{F}_q$. Let $\boldsymbol{G}$ be a generator matrix of $\mathcal{C}$ with columns $\boldsymbol{c}_1, \boldsymbol{c}_2, \ldots, \boldsymbol{c}_n$. A set $R \subseteq \{1, 2, \dots,n\}$  is a recovery set for a coordinate $i \notin R$ if $\boldsymbol{c}_i \in \langle \boldsymbol{c}_j  :  j\in R \rangle$, the linear space spanned by  $\{\boldsymbol{c}_j \; : \; j\in R\}$. In this case, for any codeword $\boldsymbol{x}\in \mathcal{C}$,  $x_i$ can be obtained from the coordinates $x_j$,  $j\in R$, just by solving the linear system whose  augmented matrix is $(\boldsymbol{c}_j, j\in R\, | \, \boldsymbol{c}_i)$.

Let $R$ be a set of cardinality $\#R=r$ and let $\pi_R:\mathbb{F}_q^n\rightarrow \mathbb{F}_q^r$ be the projection on the coordinates in $R$. For $\boldsymbol{x} \in \mathbb{F}_q^n$ we write $\boldsymbol{x}_R = \pi_R (\boldsymbol{x})$. Often we shall consider the punctured  and shortened codes:
\[
\mathcal{C}[R] := \{\boldsymbol{x}_R : \boldsymbol{x} \in \mathcal{C}\} \mbox{ and }
\mathcal{C}[[R]] := \{\boldsymbol{x}_R : \boldsymbol{x} \in \mathcal{C}, \mbox{supp}(\boldsymbol{x}) \subseteq R\},
\]
respectively, where $\mbox{supp}(\boldsymbol{x})$ denotes the {\em support} of $\boldsymbol{x}$,   $\mbox{supp}(\boldsymbol{x}):=\{ i  :  x_i\neq 0\}$. Note that $\boldsymbol{c}_i \in \langle \boldsymbol{c}_j  :  j\in R \rangle$ if and only if $\dim(\mathcal{C}[R]) = \dim(\mathcal{C}[\overline{R}])$. So the notion of recovery set does not depend on the generator matrix chosen.  If $\boldsymbol{c}_i \in \langle \boldsymbol{c}_j  :  j\in R \rangle$,  there  exist  $w_1, w_2, \ldots, w_n\in \mathbb{F}_q$ such that $\sum_{j=1}^n w_j\boldsymbol{c}_j = 0$ with $w_i \neq 0$ and $w_j = 0$ if  $j\notin \overline{R}$.  Then $\boldsymbol{w} = (w_1, w_2, \ldots, w_n) \in \mathcal{C}^{\perp}$, the dual of $\mathcal{C}$,  and $\boldsymbol{w}_R \in \mathcal{C}^{\perp}[[R]]$. Thus  $R$ is a recovery set for the coordinate $i$ if and only if there exists a word $\boldsymbol{w}_R \in \mathcal{C}^{\perp}[[R]]$ with $w_i \neq 0$. In this case $\# R \ge d(\mathcal{C}^{\perp}) -1$.

The smallest cardinality of a recovery set $R$ for a coordinate $i$ is the locality of $i$. The locality of $\mathcal{C}$, often denoted by $r=r(\mathcal{C})$, is the largest locality of any of its coordinates. Thus, we have proved the following result.

\begin{pro}
\label{d1dual}
The locality $r$ of an LRC code $\mathcal{C}$ satisfies $r\ge d(\mathcal{C}^{\perp})-1$.
\end{pro}

A code $\mathcal{C}$ reaching equality in the bound given by Proposition \ref{d1dual} will be called {\em sharp}. Note that all cyclic codes are sharp. Apart from Proposition \ref{d1dual}, perhaps the most important bound on the locality $r$ of an {LCR} code with parameters $[n,k,d]$ is given by the following Singleton-like inequality, see \cite{GHSY}.

\begin{teo} \label{Singd1}
The locality $r$ of an LRC code $\mathcal{C}$ satisfies
\begin{equation}\label{Singd1Eq}
d+k+\left\lceil \frac{k}{r}\right\rceil \le n+2.
\end{equation}
\end{teo}

The difference between the two terms in Theorem \ref{Singd1}, $D_1:=n+2-d-k-\lceil k/r\rceil$, is the {\it LRC-Singleton  defect} of $\mathcal{C}$. Codes with $D_1=0$ are called {\em Singleton-optimal} (or simply {\em optimal}). While optimal LRC codes are known for all lengths $n\le q$, \cite{KTBY},  the problem of finding codes of this type when $n>q$ is currently a challenge \cite{GXY}.
To avoid confusion in what follows, we shall sometimes refer to $r$ as the {\em classical} locality of $\mathcal {C}$.

The LRC codes that we have described above allow local recovery of the information stored in a failed node.  However,  concurrent failures of several nodes in a network are also possible and uncommon. This problem was first treated in \cite{PGLK}. According to the definition given in that article, an LRC code $\mathcal{C}$ has locality $(r,\delta)$ if for any coordinate $i$ there exists a set of positions $\overline{R}=\overline{R}(i)\subset \{1, 2, \dots,n\}$ such that

\indent (RD1) $i\in \overline{R}$  and  $\# \overline{R}\le r+\delta-1$; and\newline
\indent (RD2) $d(\mathcal{C}[\overline{R}])\ge \delta$.

The sets $\overline{R}$ satisfying the above conditions (RD1) and (RD2) are called $(r,\delta)$ recovery sets. Given such a set $\overline{R}$ and $i\in\overline{R}$, the correction capability of $\mathcal{C}[\overline{R}]$ can be used to correct an erasure at position $i$ plus any other $\delta-2$ erasures in $\overline{R}\setminus\{i \}$.
Notice that the original definition of locality of LRC codes corresponds to the case $\delta = 2$.
Provided that $\delta\ge 2$, any subset $R \subset \overline{R}$ of cardinality $r$ with $i \notin R$, satisfies $d(\mathcal{C}([R]\cup\{ i\}))\ge 2$ and consequently $R$ is a recovery set for $i$. Thus  if $\mathcal{C}$ has a locality $(r,\delta)$, then  the classical locality of $\mathcal{C}$ is $\le r$ and the number of recovery sets of cardinality $r$ for any coordinate $i$ is at least
$$
\binom{\# \overline{R}-1}{r},
$$
which can be relevant to improve the availability of $\mathcal{C}$ for recovering erasures. We remark that associated to $\mathcal{C}$   we have  several localities $(r,\delta)$, corresponding to the $d(\mathcal{C})-1$ values of $\delta=2, 3, \dots,d(\mathcal{C})$. These localities satisfy
the following generalization of the Singleton-like bound of Theorem \ref{Singd1}, which was proved in \cite{PGLK}.

\begin{pro} \label{Singdelta}
Let $\mathcal{C}$ be an LRC code with parameters $[n,k,d]$ and locality $(r,\delta)$. Then, the following inequality holds
\begin{equation} \label{SingdeltaEq}
d+k+\left( \left\lceil \frac{k}{r}\right\rceil-1\right) (\delta-1) \le n+1.
\end{equation}
\end{pro}

Analogously to what was done for the classical locality $r$, for  $t=1, 2, \dots, d(\mathcal{C})-1$, in this article we define
\begin{align*}
r_t  =r_t(\mathcal{C}):=&\min \big\{  \rho : \mbox{ for all $i=1, 2, \dots,n$, there is a set $\overline{R}_i \subseteq \{1,2, \ldots, n\}$} \nonumber \\
       &  \qquad {} \mbox{with $i \in \overline{R}_i$,  $\# \overline{R}_i\le \rho$ and $d(\mathcal{C}[\overline{R}_i])\ge t+1$}     \big\}-1.
\end{align*}

The value $r_t$ is the minimum number of positions, $\# \overline{R}-1$, needed to recover a given coordinate $i\in \overline{R}$ of any codeword $\boldsymbol{x}$, when at most $t$ erasures occur in $\boldsymbol{x}_{\overline{R}}$. Clearly $r_1$ is the classical locality of $\mathcal{C}$. We refer to $r_t$ as the {\em $t$-locality} of $\mathcal{C}$.  For example, since puncturing $< d$ times an MDS code gives a new MDS code of the same dimension, for $t<d$ the $t$-locality of an $[n, k, d]$ MDS code is $r_t = k + t-1$.

Note that from the above definitions, the code $\mathcal{C}$ has locality $(\rho,\delta)$ if and only if $r_{\delta-1}\le \rho+\delta-2$. Thus we can translate the bound given by Proposition \ref{Singdelta} in terms of $r_t$'s, as
\begin{equation} \label{Singdt}
d+k+ \left\lceil \frac{k}{r_t-t+1}\right\rceil t \le n+t+1.
\end{equation}
The difference between the two terms of Inequation (\ref{Singdt})
\begin{equation}
\label{AA}
D_t:=n+t+1-d-k- \left\lceil \frac{k}{r_t-t+1} \right\rceil t,
\end{equation}
is the $t$-th LRC-Singleton defect of $\mathcal{C}$. Codes with $D_t=0$ will be called  $t$-optimal. For example, MDS codes are $t$-optimal for all $t=1, 2, \dots,d-1$.

The sequence $(r_1, r_2, \ldots,r_{d-1})$ we have associated to an LRC code $\mathcal{C}$, resembles, up to some extent, the weight hierarchy of $\mathcal{C}$. Let us recall that for $t=1, 2, \ldots, k=\dim(\mathcal{C})$, the $t$-th {\em generalized Hamming weight of} $\mathcal{C}$ is defined as
$$
d_t=d_t(\mathcal{C}):=\min \{  \# \mbox{supp}(E) \, : \, \mbox{$E$ is a $t$-dimensional subcode of $\mathcal{C}$}\},
$$
where
$
\mbox{supp}(E):= \{ i \, : \, \mbox{there exists $\boldsymbol{x}\in E$ with $x_i\neq 0$} \},
$
see \cite[Section 3.3]{PWBJ}. We extend the bound given by Proposition \ref{d1dual} to all localities $r_t$'s  in the following new result.

\begin{pro} \label{dtdual}
For $t=1, 2, \dots, d-1$, the $t$-locality of an $[n,k,d]$ LRC code $\mathcal{C}$ satisfies $r_t\ge d_t(\mathcal{C}^{\perp})-1$, where $d_t(C^{\perp})$ is the $t$-th generalized Hamming weight of the dual code $ \mathcal{C}^{\perp}$.
\end{pro}
\begin{proof}
First note that $d-1\le \dim(\mathcal{C}^{\perp})$.  Let $\overline{R}$ be a set of coordinates such that $\# \overline{R}\le r_t+1$ and $d(\mathcal{C}[\overline{R}])\ge t+1$. According to the Singleton bound, we have $\dim(C[\overline{R}])\le \# \overline{R}-t$. Since $\mathcal{C}[\overline{R}]^{\perp}=\mathcal{C}^{\perp}[[\overline{R}]]$ (see \cite{PWBJ}, Proposition 3.1.17), it holds that $\dim(\mathcal{C}^{\perp}[[\overline{R}]])\ge t$. Thus  $d_t(\mathcal{C}^{\perp}) \le \#\overline{R}$ and the result follows.
\end{proof}

The above result can be stated in terms of localities $(r,\delta)$ as follows.

\begin{cor}
Let $\mathcal{C}$ be an LRC code with locality $(r,\delta)$. Then the following inequality holds
\[
r+ \delta \geq d_{\delta -1} (\mathcal{C}^\perp) +1.
\]
\end{cor}
\begin{proof}
From the definition of locality $(r,\delta)$ we have $r_{\delta-1} \leq r+\delta-2$. Since $r_t<r_{t+1}$ for all $1 \leq t \leq d-2$, we deduce $r_{\delta-1} = r+\delta-2$. Then,  Proposition \ref{dtdual} gives $r+ \delta \geq d_{\delta -1} (\mathcal{C}^\perp) +1$.
\end{proof}

\section{$J$-affine variety codes giving LRC codes}
\label{secdos}

In this section we  show that subfield-subcodes of some  codes arising from $J$-affine varieties are LRC codes with good recovery properties. We keep the notations as in the previous sections. In particular our LRC codes will be defined over the finite field $\mathbb{F}_q$, where $q=p^s$ and $p$ is a prime number. We shall consider an extension field $\mathbb{F}_Q$ of $\mathbb{F}_q$, where $Q=p^\ell$ and $s$ divides $\ell$. The affine varieties we manage, and so the codes arising from them, will be defined over $\mathbb{F}_Q$.  Subfield-subcodes of these codes will be defined over $\mathbb{F}_q$.

The concept of $J$-affine variety code  was introduced in \cite{QINP} and used in \cite{QINP2, LCD}  for constructing quantum and LCD codes with good parameters.  In the first subsection we recall
the construction of $J$-affine variety codes over $\mathbb{F}_Q$ and their subfield-subcodes over $\mathbb{F}_q$.

\subsection{$J$-affine variety codes and their subfield-subcodes}
\label{afine}

Let $\mathbb{F}_q$,  $q=p^s$,  be a finite field and
let $\mathbb{F}_Q$,  $Q=p^\ell$,  be an extension field of $\mathbb{F}_q$.
Let $\mathfrak{R}:= \mathbb{F}_Q[X_1, X_2, \ldots,X_m]$ be the polynomial ring in $m \geq 1$ variables over  $\mathbb{F}_Q$.  For simplicity we will often write the monomial
 $X_1^{a_1} X_2^{a_2}\cdots X_{m}^{a_m}\in \mathfrak{R}$  as $X^{\boldsymbol{a}}$, with  $\boldsymbol{a}=(a_1,a_2, \ldots, a_m) $.
Fix positive integers $N_j>1$, $j=1,2,\dots, m$, such that  $N_j-1$ divides $Q-1$. Let $J$ be a subset of indices of  variables, $J\subseteq \{1,2, \ldots, m\}$, and let  $I_J$ be the ideal of
$\mathfrak{R}$ generated by the binomials
$X_j^{N_j -1} - 1$ if $j \in J$, and $X_j^{N_j} - X_j$ if $j \not \in J$.  Denote by $\mathfrak{R}_J$ the
quotient ring $\mathfrak{R}_J=\mathfrak{R}/I_J$. Set $T_j = N_j -2$  if $j \in J$ and $T_j = N_j -1$ otherwise, and let
$$
\mathcal{H}_J  : = \{0,1,\ldots,T_1\}\times \{0,1,\ldots,T_2\} \times\cdots\times\{0,1,\ldots,T_m\}.
$$
Let $Z_J = \{P_1, P_2, \ldots, P_{n_J}\}$ be the set of zeros of $I_J$ over $\mathbb{F}_Q$. This set has cardinality
$$
n_J = \prod_{j \notin J} N_j \prod_{j \in J} (N_j -1).
$$
Consider the  well-defined evaluation map
$$
\mathrm{ev}_J: \mathfrak{R}_J \rightarrow \mathbb{F}_{Q}^{n_J} \; ,  \;
\mathrm{ev}_J(f) = (f(P_1), f(P_2),\ldots, f(P_{n_J})),
$$
where $f$ denotes both  the polynomial in $\mathfrak{R}$ and its corresponding equivalence class in $\mathfrak{R}_J$.

\begin{de}\label{def:unouno}
{\rm Given a non-empty subset $\Delta\subseteq \mathcal{H}_J$, the {\it $J$-affine variety code $E^J_\Delta$,} is the linear  subspace $E^J_\Delta:=\langle \mathrm{ev}_J (X^{\boldsymbol{a}}) \, : \, \boldsymbol{a} \in \Delta \rangle \subseteq   \mathbb{F}_Q^{n_J}$.}
\end{de}

Then $E^J_\Delta$ is a linear code over $\mathbb{F}_Q$. Its length is $n_J$ and its dimension equals the cardinality of $\Delta$, since $\mathrm{ev}_J$ is injective, \cite{QINP}.  Recall that $q=p^s$ where $s$ divides $\ell$ and thus $\mathbb{F}_q$ is a subfield of $\mathbb{F}_Q$.

\begin{de}\label{def:unodos}
{\rm  The {\it  subfield-subcode} of $E^J_\Delta$ over the field $\mathbb{F}_q$, denoted $\mathcal{C}_\Delta^{J}$,  is the linear code $\mathcal{C}_\Delta^{J}:= E_\Delta^J \cap \mathbb{F}_{q}^{n_J}$.
}
\end{de}

In order to study the codes $\mathcal{C}_\Delta^{J}$,
we shall manage  the elements of $\mathcal{H}_J$   in a particular manner. Let $j$,  $1 \leq j \leq m$. If $j \in J$ then we identify the set  $\{0,1, \ldots, T_j\}$ with   the ring $\mathbb{Z}/(T_j +1) \mathbb{Z}$.
When $j \not \in J$,   we identify the set  $\{1, 2,\ldots, T_j\}$ with $\mathbb{Z}/T_j  \mathbb{Z}$, and we extend the addition and multiplication of this ring to $\{0,1, \ldots, T_j\}$, by setting $0+\alpha=\alpha$, $0\cdot\alpha=0$ for all $\alpha=0,1,\dots,T_j$. The reason that explains these different ways of treating $\{0,1, \ldots, T_j\}$ is the fact that the evaluation of monomials containing $X_j^0$ or containing $X_j^{N_j-1}$  may be different when $j\not\in J$, see \cite{QINP2} for details.

Under the above conventions,
a set $\mathfrak{S}\subseteq \mathcal{H}_J$  is a {\it cyclotomic set  with respect to $q$} if $q \boldsymbol{y} \in \mathfrak{S}$ for all $\boldsymbol{y} = (y_1, y_2, \ldots, y_m) \in \mathfrak{S}$.  Minimal cyclotomic sets are those of the form $\mathfrak{I}=\{ q^{i } \boldsymbol{y} \, : \, i \geq 0\}$, for some element $\boldsymbol{y} \in \mathcal{H}_J$.
For each minimal cyclotomic set $\mathfrak{I}$,  we consider a unique representative $\boldsymbol{a} = (a_1, a_2, \ldots, a_m)\in \mathfrak{I}$,  constructed iteratively as follows:
$a_1=\min\{ y_1 :  (y_1, y_2, \ldots, y_m) \in \mathfrak{I} \}$, and
$a_j=\min\{ y_j :  (a_1, a_2, \ldots,a_{j-1},y_j,\ldots, y_m) \in \mathfrak{I} \}$ for $j=2,3,\dots,m$. We shall denote by $\mathfrak{I}_{\boldsymbol{a}}$ the minimal cyclotomic set with representative $\boldsymbol{a}$ and by $i_{\boldsymbol{a}}$ the cardinality of $\mathfrak{I}_{\boldsymbol{a}}$. Thus $\mathfrak{I}_{\boldsymbol{a}}=\{ \boldsymbol{a},q\boldsymbol{a},\dots,q^{(i_{\boldsymbol{a}} -1)}\boldsymbol{a} \}$.

Let $\mathcal{A}$ be the set of representatives of all minimal cyclotomic sets in $\mathcal{H}_J$. Given a non-empty subset $\Delta\subseteq\mathcal{H}_J$, we define
$\mathcal{A}(\Delta)=\{  \boldsymbol{a} \in \mathcal{A} \, : \, \mathfrak{I}_{\boldsymbol{a}} \subseteq \Delta \}$.
The set $\Delta$ is called {\em closed} if it is a union of minimal cyclotomic sets, that is, if
$$
\Delta=\bigcup_{\boldsymbol{a} \in\mathcal{A}(\Delta)} \mathfrak{I}_{\boldsymbol{a}}.
$$

An important tool to study subfield-subcodes is the trace map.
Since we are interested in subfield-subcodes over $\mathbb{F}_{q}$ of evaluation codes over $\mathbb{F}_{Q}$, for $\boldsymbol{a} \in \mathcal{A}$ we consider the map
$$
\mathcal{T}_{\boldsymbol{a} }: \mathfrak{R}_J \rightarrow \mathfrak{R}_J \; , \;
\mathcal{T}_{\boldsymbol{a}} (f) = f + f^{q} + \cdots + f^{q^{(i_{\boldsymbol{a}} -1)}}.
$$
Let $\xi_{\boldsymbol{a}}$ be a fixed primitive element of the field $\mathbb{F}_{q^{ i_{\boldsymbol{a}}}}$. The next result gives an explicit description of the code  $\mathcal{C}_\Delta^{J}$. It   extends Theorem 4 in \cite{galindo-hernando}.  Here we state the result for any set $J \subseteq \{1,2, \ldots, m\}$, while in \cite{galindo-hernando} only the case $J=\{1,2, \ldots, m\}$ was considered.

\begin{teo}
\label{ddimension}
With the above notation, if $\Delta\subseteq\mathcal{H}_J$  then the set of vectors
$$
\bigcup_{\boldsymbol{a} \in \mathcal{A}(\Delta)}
\left\{  \mathrm{ev}_J(\mathcal{T}_{\boldsymbol{a}} (\xi_{\boldsymbol{a}}^{k} X^{\boldsymbol{a}})) \, : \,  0 \leq k \leq i_{\boldsymbol{a}} -1 \right\}
$$
is a basis of $\mathcal{C}_\Delta^{J}$ over $\mathbb{F}_{q}$. In particular, if $\Delta$ is a closed set, then  $\dim(\mathcal{C}_\Delta^{J})=\dim(E_\Delta^{J})=\#\Delta$.

\end{teo}

The proof of Theorem \ref{ddimension}  is similar to that of Theorem 4 in \cite{galindo-hernando}
and we omit it. Instead we show an example illustrating this theorem.

\begin{exa}{\rm
Take $p=2$, $s=3$, $\ell=6$ and $m=2$, so  $q=2^3=8$ and $Q= 2^6=64$. Take $J=\{1\}$, $N_1=8$ and $N_2=10$, so that $T_1=6$ and $T_2=9$.
Let $\boldsymbol{a}_1=(1,2), \boldsymbol{a}_2=(2,3)$ and $\boldsymbol{a}_3=(1,3)$.
Then $\mathfrak{I}_{\boldsymbol{a}_1}=\{(1,2),(1,7)\}, \mathfrak{I}_{\boldsymbol{a}_2}=\{(2,3),(2,6)\}$ and
$\mathfrak{I}_{\boldsymbol{a}_3}=\{(1,3),(1,6)\}$, hence $i_{\boldsymbol{a}_1}= i_{\boldsymbol{a}_2}=i_{\boldsymbol{a}_3}=2$.
Let $\Delta_1 = \mathfrak{I}_{\boldsymbol{a}_1} \cup \mathfrak{I}_{\boldsymbol{a}_2}$ and
$\Delta_2 = \mathfrak{I}_{\boldsymbol{a}_1} \cup \mathfrak{I}_{\boldsymbol{a}_2} \cup \{ (1,3)\}$. Thus $\Delta_1$ is closed but $\Delta_2$ is not.
Consider the affine variety codes $E_{\Delta_1}^J, E_{\Delta_2}^J$ defined over $\mathbb{F}_{64}$  and the subfield-subcodes
$\mathcal{C}_{\Delta_1}^J, \mathcal{C}_{\Delta_2}^J$ over $\mathbb{F}_{8}$.  All of them have length $n_J=70$. Furthermore
$\dim(E_{\Delta_1}^J)=4, \dim(E_{\Delta_2}^J)=5$. In fact we have
\begin{eqnarray*}
E_{\Delta_1}^J&=&\langle
\mathrm{ev}_J(X_1 X_2^2),  \mathrm{ev}_J(X_1 X_2^7), \mathrm{ev}_J( X_1^2 X_2^3), \mathrm{ev}_J(X_1^2 X_2^6)\rangle, \\
E_{\Delta_2}^J&=&\langle
\mathrm{ev}_J(X_1 X_2^2),  \mathrm{ev}_J(X_1 X_2^7), \mathrm{ev}_J( X_1^2 X_2^3), \mathrm{ev}_J(X_1^2 X_2^6),  \mathrm{ev}_J(X_1 X_2^3)\rangle
\end{eqnarray*}
over $\mathbb{F}_{64}$.  And from Theorem \ref{ddimension}
\begin{multline*}
\mathcal{C}_{\Delta_1}^J = \mathcal{C}_{\Delta_2}^J = \langle
\mathrm{ev}_J(X_1 X_2^2 + X_1 X_2^7),  \mathrm{ev}_J(\xi X_1 X_2^2 + \xi^8 X_1 X_2^7), \\
 \mathrm{ev}_J( X_1^2 X_2^3 + X_1^2 X_2^6),\mathrm{ev}_J( \xi X_1^2 X_2^3 + \xi^8 X_1^2 X_2^6) \rangle
\end{multline*}
over the field $\mathbb{F}_{8}$, where $\xi$ a primitive element of $\mathbb{F}_{64}$.
This example shows that when we  study the properties of a code $\mathcal{C}_{\Delta}^J $, we can always assume that the set $\Delta$, from which it arises, is closed.}
\end{exa}

\subsection{LRC codes from $J$-affine variety codes}
\label{subsect2}

In this subsection we present some specific families of $J$-affine variety codes whose subfield-subcodes are LRC codes. We determine  recovery sets for these LRC codes and show that their localities $(r, \delta)$  behave well.

Let us remember that the construction of $J$-affine variety codes begins by taking a set of indices $J\subseteq \{ 1,2,\dots,m\}$ and integers $N_1,N_2,\dots,N_m$, such that $N_j-1$ divides $Q-1$ for all $j$. In order to obtain good LRC codes, from now on we shall assume the additional property that $N_1,N_2,\dots,N_m$ have been chosen in a way that there exists a non-empty subset $L\subseteq J$ such that $q -1$ divides $N_j -1$ for all $j\in L$.
Throughout the rest of this section we shall assume that the integers $N_1,N_2,\dots,N_m$, and the sets $J$ and $L$, have been fixed satisfying the above conditions.

Let $\alpha$ and $\eta$ be  primitive elements of $\mathbb{F}_{Q}$ and $\mathbb{F}_{q}$, respectively. For $1 \leq j \leq m$,  let $\gamma_j = \alpha^{(Q-1)/(N_j -1)} \in \mathbb{F}_{Q}$. The following property will be used later.

\begin{lem}
\label{lema16}
Let $l$ and $n$ be two nonnegative integers.  If $j \in L$,  then the following equality holds in $\mathbb{F}_{Q}$,
\[
\left(\gamma_j^{l} \eta^n \right)^{N_j -1} = 1.
\]
\end{lem}

\begin{proof}
The statement follows from the chain of equalities
\[
\left(\gamma_j^{l} \eta^n \right)^{N_j -1} = \left(\alpha^{\frac{Q-1}{N_j -1}}\right)^{l (N_j -1)} \left(\eta^{N_j -1}\right)^n = \left(\alpha^{Q-1}\right)^l \left(\eta^{q -1}\right)^{n \frac{N_j-1}{q-1}}=1.
\]
\end{proof}

As defined in Subsection \ref{afine}, let $I_J$ be the ideal in $\mathfrak{R}$ generated by the binomials
$X_j^{N_j -1} - 1$ if $j \in J$ and $X_j^{N_j} - X_j$ if $j \not \in J$, and let
$Z_J = \{P_1, P_2, \ldots, P_{n_J}\}$ be the set of zeros of $I_J$ over $\mathbb{F}_Q$.
In this subsection we determine recovery sets for codes $\mathcal{C}_\Delta^{J}$. These recovery sets will be obtained from subsets $R\subset Z_J$ satisfying some geometrical properties. Given a point $P\in Z_J$ we set $\mbox{coord}(P):=j$ if $P=P_j$; consequently, given a set $R\subseteq Z_J$, we set $\mbox{coord}(R):=\{ \mbox{coord}(P) : P\in R\}$.

Given a nonzero element $\lambda \in \mathbb{F}_{Q}^*$ and a point $P\in \mathbb{F}_Q^m$, we define the product $\lambda \cdot_L P$ as the point of $\mathbb{F}_{Q}^m$ obtained by multiplying by $\lambda$ the coordinates of $P$ corresponding to positions in $L$ and leaving unchanged its remaining coordinates.

\begin{lem}
\label{lema16bis}
If $P\in Z_J$, then $\eta^n \cdot_L P \in Z_j$ for every nonnegative integer $n$.
\end{lem}

The proof of this lemma follows directly from the definition of $L$ and Lemma \ref{lema16}.
We define the {\em orbit} of a point $P_{t_0}\in Z_J$ as the set
\begin{equation}
\label{pi}
R_{t_0} := \{  \eta^n \cdot_L P_{t_0} \, : \,  0  \leq n \leq q -2 \}.
\end{equation}
Notice that $R_{t_0} \subset Z_j$  by Lemma \ref{lema16bis}. As we shall see later, these orbits are closely related to recovery sets of our codes. For short, the point $\eta^n \cdot_L P_{t_0}$ will be denoted $P^L_{n,t_0}$.

Let $\mathcal{A}$ be the set of representatives of all minimal cyclotomic sets in $\mathcal{H}_J$, as defined in Subsection \ref{afine}. For $\boldsymbol{a}\in \mathcal{A}$ we write $\sigma_L(\boldsymbol{a})=\sum_{j \in L} a_j$, where the $a_j$'s and the sum $\sigma_L(\boldsymbol{a})$ are seen as integers.

\begin{lem}
\label{lema17}
Let $\boldsymbol{a}\in\mathcal{A}$ and let $k$ and $n$ be two integers such that $0 \leq k \leq i_{\boldsymbol{a}} -1$ and $0 \leq n \leq q-1$. Then we have
\begin{equation}
\label{el17}
\mathcal{T}_{\boldsymbol{a}} (\xi_{\boldsymbol{a}}^{k} X^{\boldsymbol{a}}) \left( P^L_{n,t_0} \right) = \eta^{n \sigma_L (\boldsymbol{a})} \mathcal{T}_{\boldsymbol{a}} (\xi_{\boldsymbol{a}}^{k} X^{\boldsymbol{a}}) \left( P_{t_0}\right).
\end{equation}
\end{lem}
\begin{proof}
Since no coordinate of  $P_{t_0}$ in the positions of $L$ vanishes, we can write $P_{t_0} = (\gamma_1^{k_1}, \gamma_2^{k_2}, \ldots, \gamma_m^{k_m})$ without loss of generality.  Then
\begin{multline}
\mathcal{T}_{\boldsymbol{a}} (\xi_{\boldsymbol{a}}^{k} X^{\boldsymbol{a}}) \left(  P^L_{n,t_0} \right)
 = \sum_{t=0}^{i_{\boldsymbol{a}}-1} \left(\xi_{\boldsymbol{a}}^{k} \prod_{l \in L} (\eta^n \gamma_l^{k_l})^{a_l} \prod_{l \not \in L} ( \gamma_l^{k_l})^{a_l} \right)^{t q}
\\
= \eta^{n \sigma_L (\boldsymbol{a})}  \sum_{t=0}^{i_{\boldsymbol{a}}-1} \left(\xi_{\boldsymbol{a}}^{k} \prod_{l=1}^m ( \gamma_l^{k_l})^{a_l}\right)^{t q} = \eta^{n \sigma_L (\boldsymbol{a})} \mathcal{T}_{\boldsymbol{a}} (\xi_{\boldsymbol{a}}^{k} X^{\boldsymbol{a}}) \left( P_{t_0} \right)
\end{multline}
as stated.
\end{proof}

The $J$-affine variety code $E^J_\Delta$ was defined as the linear subspace  spanned by the vectors $\mathrm{ev}_J (X^{\boldsymbol{a}})$,  $\boldsymbol{a} \in \Delta$, where $\Delta$ is any non-empty subset of $\mathcal{H}_J$.
Taking advantage of Theorem \ref{ddimension}, from now on all the sets $\Delta$ we consider will be closed, that is a union of minimal cyclotomic sets, $\Delta = \cup_{l=1}^r \mathfrak{I}_{\boldsymbol{a}_l}$, with $\boldsymbol{a}_l\in\mathcal{A}$, $1\le l\le r$. Later in this article we shall impose even more restrictive conditions.

\begin{teo}
\label{el21}
Let $\Delta = \cup_{l=1}^r \mathfrak{I}_{\boldsymbol{a}_l}$, where $\{ \boldsymbol{a}_1,\boldsymbol{a}_2,\dots, \boldsymbol{a}_r\}$ is a subset of $\mathcal{A}$  with cardinality $r \leq q-2$. If the integers  $\sigma_L(\boldsymbol{a}_1),\sigma_L(\boldsymbol{a}_2),\dots, \sigma_L(\boldsymbol{a}_r)$ are pairwise different modulo   $q -1$, then the subfield-subcode  $\mathcal{C}_\Delta^{J}$ is an LRC code with locality $\le r$.
\end{teo}

\begin{proof}
Let $\boldsymbol{c} =\mathrm{ev}_J(h)$ be a codeword of $\mathcal{C}_\Delta^{J}$. By Theorem \ref{ddimension}, $h$ can be written as
\[
 h = h_{\boldsymbol{a}_1} + h_{\boldsymbol{a}_2} + \cdots + h_{\boldsymbol{a}_r},
 \]
where each $h_{\boldsymbol{a}_l}$ is a linear combination of polynomials of the form $\mathcal{T}_{\boldsymbol{a}_l} (\xi_{\boldsymbol{a}_l}^{k} X^{\boldsymbol{a}_l})$, $0 \leq k \leq i_{\boldsymbol{a}_l} -1$,
and coefficients in $\mathbb{F}_{q}$. Fix  a possition $t_0\in \{ 1,2,\dots,n\}$.
We shall show that the set $R=\{P^L_{n_i, t_0} \, : \, i=1,2,\dots,r\}$ of points corresponding to $r$ consecutive nonzero $n_i$'s, gives a recovery set $\mbox{coord}(R)$ for $t_0$. According to Lemma  \ref{lema16bis}, the points in $R$ belong to $Z_J$.   Assume we know the $r$ coordinates $ h\left(P^L_{n_i, t_0}\right)$, $i=1,2,\dots,r$, of $\boldsymbol{c}$.  By linearity
\begin{equation}
\label{de21}
h\left(P_{t_0}\right) = h_{\boldsymbol{a}_1} \left(P_{t_0}\right)+ h_{\boldsymbol{a}_2} \left(P_{t_0}\right)+ \cdots + h_{\boldsymbol{a}_r}\left(P_{t_0}\right).
\end{equation}
Then, from Lemma \ref{lema17} we get the equalities
\begin{align*}
h\left(\eta^{n_1} \cdot_L P_{t_0}\right) &= \eta^{n_1 \sigma_L (\boldsymbol{a}_1)} h_{\boldsymbol{a}_1} \left(P_{t_0}\right)+ \eta^{n_1 \sigma_L (\boldsymbol{a}_2)} h_{\boldsymbol{a}_2} \left(P_{t_0}\right)+ \cdots + \eta^{n_1 \sigma_L (\boldsymbol{a}_r)} h_{\boldsymbol{a}_r}\left(P_{t_0}\right), \\
h\left(\eta^{n_2} \cdot_L P_{t_0}\right) &= \eta^{n_2 \sigma_L (\boldsymbol{a}_1)} h_{\boldsymbol{a}_1} \left(P_{t_0}\right)+ \eta^{n_2 \sigma_L (\boldsymbol{a}_2)} h_{\boldsymbol{a}_2} \left(P_{t_0}\right)+ \cdots + \eta^{n_2 \sigma_L (\boldsymbol{a}_r)} h_{\boldsymbol{a}_r}\left(P_{t_0}\right),\\
 & \vdots\\
h\left(\eta^{n_r} \cdot_L P_{t_0}\right) &= \eta^{n_r \sigma_L (\boldsymbol{a}_1)} h_{\boldsymbol{a}_1} \left(P_{t_0}\right)+ \eta^{n_r \sigma_L (\boldsymbol{a}_2)} h_{\boldsymbol{a}_2} \left(P_{t_0}\right)+ \cdots + \eta^{n_r \sigma_L (\boldsymbol{a}_r)} h_{\boldsymbol{a}_r}\left(P_{t_0}\right).
\end{align*}
Write $\eta_i := \eta^{\sigma_L (\boldsymbol{a}_i)}$, $1 \leq i \leq r$. We have obtained the square system of linear equations
\[
\left(
    \begin{array}{ccc}
      \eta_1^{n_1} & \cdots & \eta_r^{n_1} \\
    \eta_1^{n_2} & \cdots & \eta_r^{n_2} \\
      \vdots & \ddots & \vdots \\
      \eta_1^{n_r} & \cdots & \eta_r^{n_r} \\
    \end{array}
  \right)
  \left(
    \begin{array}{c}
      h_{\boldsymbol{a}_1} \left(P_{t_0}\right) \\
      h_{\boldsymbol{a}_2} \left(P_{t_0}\right) \\
      \vdots \\
      h_{\boldsymbol{a}_r} \left(P_{t_0}\right) \\
    \end{array}
  \right)=
  \left(
    \begin{array}{c}
      h\left(\eta^{n_1} \cdot_L P_{t_0}\right) \\
      h\left(\eta^{n_2} \cdot_L P_{t_0}\right) \\
      \vdots \\
      h\left(\eta^{n_r} \cdot_L P_{t_0}\right) \\
    \end{array}
  \right).
\]
The matrix of this system is of Vandermonde type, and thus nonsingular.
Then the solution is unique and gives the values $h_{\boldsymbol{a}_i} \left(P_{t_0}\right)$, $1 \leq i \leq r$. Once these values are known, from (\ref{de21}) we can deduce $h\left(P_{t_0}\right)$.
\end{proof}

Under some supplementary conditions we can obtain LRC codes with larger dimension.
In the next theorem we restrict to the case $L=\{1\}$, so that $\sigma_L(\boldsymbol{a}_l)$ is the first coordinate of $\boldsymbol{a}_l$. We denote by $a_l$ such first coordinate.

\begin{teo}
\label{el25}
Let $L=\{1\}$.
Let $\Delta = \cup_{l=1}^{q-1} \mathfrak{I}_{\boldsymbol{a}_l}$, where $\{ \boldsymbol{a}_1,\boldsymbol{a}_2,\dots, \boldsymbol{a}_{q-1}\}$ is a subset of $\mathcal{A}$ such that the first coordinates  ${a}_1,{a}_2,\dots, a_{q-1}$ of $\boldsymbol{a}_1,\boldsymbol{a}_2,\dots, \boldsymbol{a}_{q-1}$, are pairwise different modulo   $q -1$. If there is an index $v$,  $1 \leq v \leq q-1$, for which the following conditions
\begin{enumerate}
\item $a_v$ divides $N_1-1$,
\item $\gcd(a_v,a_l)=1$ for all $1 \leq l \leq q-1$, $l \neq v$, and
\item $\gcd(a_v, q -1)=1$;
  \end{enumerate}
hold, then $\mathcal{C}_\Delta^{J}$ is an LRC code with locality $\le q-1+ (a_v -1)$.
\end{teo}
\begin{proof}
As in the proof of Theorem \ref{el21}, let $\boldsymbol{c} =\mathrm{ev}_J(h)$ be a codeword of $\mathcal{C}_\Delta^{J}$. Fix a coordinate $t_0$ and consider the set of points
\[
\overline{R} := \{  \eta^n \cdot_L P_{t_0} \, : \, 0  \leq n \leq q -1 \} \cup
 \{ \omega^k \eta \cdot_L P_{t_0} \, : \, 1  \leq k \leq a_v -1 \},
\]
where $\omega$ is a primitive $a_v$-root of unity.
Since $a_v$ divides $N_1-1$, then it holds that $\omega^k \eta \cdot_L P_{t_0} \in Z_J$ for all $1  \leq k \leq a_v -1$.
We shall show that  $\mbox{coord}(\overline{R}) \setminus \{ {t_0} \}$ is a recovery set for the coordinate $t_0$.
For simplicity  we can assume $v=1$. As in Theorem \ref{el21}, we can write $h$  as
\[
 h = h_{\boldsymbol{a}_1} + h_{\boldsymbol{a}_2} + \cdots + h_{\boldsymbol{a}_{q-1}},
 \]
$h_{\boldsymbol{a}_l}$ being a linear combination  with coefficients in $\mathbb{F}_{q}$ of polynomials of the form $\mathcal{T}_{\boldsymbol{a}_l} (\xi_{\boldsymbol{a}_l}^{k} X^{\boldsymbol{a}_l})$, $0 \leq k \leq i_{\boldsymbol{a}_l} -1$. So we have
\begin{equation}
\label{sumar}
\begin{aligned}
h\left(\eta \cdot_L P_{t_0}\right) &=
\eta^{a_1} h_{\boldsymbol{a}_1} \left(P_{t_0}\right)+\eta^{a_2} h_{\boldsymbol{a}_2} \left(P_{t_0}\right)+ \cdots + \eta^{a_{q-1}} h_{\boldsymbol{a}_{q-1}}\left(P_{t_0}\right), \\
h\left(\omega \eta \cdot_L P_{t_0}\right) &=
\eta^{a_1} h_{\boldsymbol{a}_1} \left(P_{t_0}\right)+ \eta^{a_2} h_{\boldsymbol{a}_2} \left(\omega \cdot_L P_{t_0}\right)+ \cdots + \eta^{a_{q-1}} h_{\boldsymbol{a}_{q-1}}\left(\omega \cdot_L P_{t_0}\right),\\
 & \vdots\\
h\left(\omega^{a_1 -1} \eta \cdot_L P_{t_0}\right) &=
\eta^{a_1} h_{\boldsymbol{a}_1} \left(P_{t_0}\right)+\eta^{a_2} h_{\boldsymbol{a}_2} \left(\omega^{a_1 -1} \cdot_L P_{t_0}\right)+  \cdots + \eta^{a_{q-1}} h_{\boldsymbol{a}_{q-1}}\left(\omega^{a_1 -1}\cdot_L P_{t_0}\right).
\end{aligned}
\end{equation}
The facts that $a_1$ divides $N_1-1$ and $\gcd(a_1,q-1)=1$ imply that none of the points $\omega^k \eta \cdot_L P_{t_0}$, $0 \leq k \leq a_1 -1$, coincide neither with $P_{t_0}$ nor among them. Adding  the equalities in (\ref{sumar}) we get a known value on the left hand. On the right hand we get
\begin{multline}
\label{ceros}
(a_1 -1) \left(\eta^{a_1} h_{\boldsymbol{a}_1} \left(P_{t_0}\right)\right) +
\eta^{a_2} h_{\boldsymbol{a}_2} \left(P_{t_0}\right) \left( 1 + \omega^{a_2} + (\omega^{a_2})^2 + \cdots + (\omega^{a_2})^{a_1 -1} \right) + \cdots  \\ +  \eta^{a_{q-1}} h_{\boldsymbol{a}_{q-1}} \left(P_{t_0}\right) \left( 1 + \omega^{a_{q-1}} + (\omega^{a_{q-1}})^2 + \cdots + (\omega^{a_{q-1}})^{a_1 -1}  \right).
\end{multline}
Since  $\gcd(a_1,a_l) =1$  for $2 \leq l \leq q-1$, it holds that
$$
1 + \omega^{a_l} + (\omega^{a_l})^2 +\cdots + (\omega^{a_l})^{a_1 -1}=0
$$
as this expression is the sum of all $a_1$-roots of unity. Then (\ref{ceros}) becomes
$(a_1 -1) \left(\eta^{a_1} h_{\boldsymbol{a}_1} \left(P_{t_0}\right)\right)$, from which we  deduce
$h_{\boldsymbol{a}_1} \left(P_{t_0}\right)$.
The polynomial $h- h_{\boldsymbol{a}_1} = h_{\boldsymbol{a}_2} + \cdots + h_{\boldsymbol{a}_{q-1}}$ is related to  $q-2$ minimal cyclotomic  sets, hence we can apply now the procedure developed in  the proof of Theorem \ref{el21} to  compute $(h - h_{\boldsymbol{a}_1}) \left(P_{t_0}\right)$. Finally   $h\left(P_{t_0}\right)=(h - h_{\boldsymbol{a}_1})\left(P_{t_0}\right) +h_{\boldsymbol{a}_1} \left(P_{t_0}\right)$.
\end{proof}

\subsection{The case of $\{1,2, \ldots,m\}$-affine variety codes}
\label{subsect3}

In this subsection we study the codes $\mathcal{C}^{J}_\Delta$ when $J$ equals the whole set of indices,  $J=\{1,2, \ldots, m\}$.  Let $L$ be a non-empty subset of $J$.
As in the previous subsection, we take the closed set $\Delta$  as a union of minimal cyclotomic sets,  $\Delta = \cup_{i=1}^z \mathfrak{I}_{\boldsymbol{a}_i}$, but now we  shall add the condition that the set $\{ \sigma_L({\boldsymbol{a}_1}),\sigma_L({\boldsymbol{a}_2}),\dots, \sigma_L({\boldsymbol{a}_z}) \}$ contains exactly $r$ consecutive integers. With these ingredients we construct the $J$-affine variety code $E^{J}_\Delta$
over $\mathbb{F}_Q$ and its subfield-subcode $\mathcal{C}^{J}_\Delta$ over $\mathbb{F}_q$.

Let us first deal with the case $m=1$ and $r=z$.  So let  $J=L=\{1\}$.
The set $\mathcal{A}$ of representatives of all minimal cyclotomic sets in $\mathcal{H}_J$
can be seen now as a subset of $\mathbb{Z}$.
Let  $\boldsymbol{a}_1,\boldsymbol{a}_2 , \dots , \boldsymbol{a}_r, \boldsymbol{a}_{r+1}$ be the $r+1$ smallest (with respect to the natural order in $\mathbb{Z}$) representatives in $\mathcal{A}$  and let $\Delta= \cup_{l=1}^{r} \mathfrak{I}_{\boldsymbol{a}_l}$.
The codes  $E_\Delta^{J}$ we obtain in this case were studied in \cite{QINP2}, where the following result was proved by using the BCH bound.

\begin{pro}
\label{nuevapro}
{\rm (\cite[Theorem 3.7]{QINP2})}
Let $m=1$ and  let $\Delta= \cup_{l=1}^{r} \mathfrak{I}_{\boldsymbol{a}_l}$, where
$\boldsymbol{a}_1,\boldsymbol{a}_2 , \dots , \boldsymbol{a}_r$, $\boldsymbol{a}_{r+1}$ are the $r+1$ smallest elements of $\mathcal{A}$.
Then the minimum distance of the dual code  $(E_\Delta^{J})^{\perp}$ satisfies $d( (E_\Delta^{J})^{\perp})\ge \boldsymbol{a}_{r+1}+1$.
\end{pro}

Let us recall that there exists a close relation between the dual of any linear code $\mathcal{D}$ of length $n$, defined over $\mathbb{F}_Q$ and the subfield-subcode $\mathcal{D}\cap \mathbb{F}_q^n$. This relation is given by the Delsarte Theorem,  as follows
$$
(\mathcal{D} \cap \mathbb{F}_q^n)^{\perp} = \mbox{Tr}(\mathcal{D}^{\perp})
$$
where $\mbox{Tr}$ is the trace map of the extension $\mathbb{F}_Q/\mathbb{F}_q$, see \cite{delsarte}.
In our case $m=1$, this theorem implies
\begin{equation}
\label{ilatina}
\left( \mathcal{C}^{J}_\Delta\right)^\perp = (E^{J}_\Delta)^\perp \cap  \mathbb{F}_q^{n_J}
\end{equation}
see \cite[Proposition 11]{IEEE} for a complete proof of this equality.
If $r\le q-2$, then the $r$ smallest elements of $\mathcal{A}$ are $\boldsymbol{a}_l=l-1$, $l=1,2,\dots,r$, and hence from Proposition \ref{nuevapro} and Equality (\ref{ilatina})  we have $d((\mathcal{C}^{J}_\Delta)^{\perp})\ge r+1$.

\begin{pro}
\label{el22a}
Let $m=1$ and  $\Delta= \cup_{l=1}^{r} \mathfrak{I}_{\boldsymbol{a}_l}$, where $r\le q-2$ and
$\boldsymbol{a}_1,\boldsymbol{a}_2, \dots ,\boldsymbol{a}_r$  are  the $r$ smallest elements of $\mathcal{A}$. Then $\mathcal{C}^{J}_\Delta$ is a sharp LRC code of locality $r_1=r$.
\end{pro}
\begin{proof}
According to Theorem \ref{el21}, $\mathcal{C}^{J}_\Delta$ is an  LRC code of locality $\le r$. Since $d((\mathcal{C}^{J}_\Delta)^{\perp})\ge r+1$, Proposition  \ref{d1dual} implies the result.
\end{proof}

Le us study now the general case $m\ge 1$. So let $J=\{ 1,2,\dots,m\}$,
let $L$ be a non-empty subset of $J$,
 and  $\Delta= \cup_{l=1}^{z} \mathfrak{I}_{\boldsymbol{a}_l}$. We are going to construct codes $\mathcal{C}^{J}_\Delta$ with locality $(r,q-r)$ for some $r\le z$, whose $(r,q-r)$ recovery sets satify the conditions  (RD1) and (RD2) stated in Section \ref{secuno} with equality.
Fix a coordinate $t_0$, $1\le t_0\le n_J$, and let us consider the orbit of $P_{t_0}$, defined in Equation (\ref{pi}),
$$
R_{t_0} = \{  \eta^n \cdot_L P_{t_0} \, : \,  0  \leq n \leq q -2 \} = \{ P_{t_0}, P^L_{1,t_0},\dots, P^L_{q-2,t_0}\} \subseteq Z_J.
$$

\begin{lem}
\label{dimens}
Let $J=\{ 1,2,\dots,m\}$ and let  $\Delta= \cup_{l=1}^{z} \mathfrak{I}_{\boldsymbol{a}_l}$
be a closed set. Let $t_0$ be a
\linebreak
coordinate, $1\le t_0\le n_J$, and let $R_{t_0}$ be the orbit of $P_{t_0}\in Z_J$.
If the set $\{ \sigma_L(\boldsymbol{a}_1),$
\linebreak
$\sigma_L(\boldsymbol{a}_2),\dots,$ $\sigma_L(\boldsymbol{a}_z) \}$
contains exactly $r\le q-2$ distinct elements modulo $q-1$,
then the punctured code $\mathcal{C}^{J}_\Delta[\mbox{\rm coord}(R_{t_0})]$ has length $q-1$ and dimension $r$.
\end{lem}
\begin{proof}
The statement about the length is clear, since $R_{t_0}\subseteq Z_J$. Let us compute the dimension of $\mathcal{C}^{J}_\Delta[\mbox{coord}(R_{t_0})]$. According to Theorem \ref{ddimension}, this code is generated by the set of vectors
\begin{equation}
\label{generators}
\boldsymbol{v}_{l k}= \left(\mathcal{T}_{\boldsymbol{a}_l} (\xi_{\boldsymbol{a}_l}^{k} X^{\boldsymbol{a}_l}) (P_{t_0} ),
\mathcal{T}_{\boldsymbol{a}_l} (\xi_{\boldsymbol{a}_l}^{k} X^{\boldsymbol{a}_l}) (P^L_{1,t_0} ), \dots,
\mathcal{T}_{\boldsymbol{a}_l} (\xi_{\boldsymbol{a}_l}^{k} X^{\boldsymbol{a}_l}) (P^L_{q-2,t_0} )\right)
\end{equation}
$1\le l\le z;\,  0 \leq k \leq i_{\boldsymbol{a}_l} -1$.
From Lemma  \ref{lema17} we have
$$
\boldsymbol{v}_{l k}= \left(\mathcal{T}_{\boldsymbol{a}_l} (\xi_{\boldsymbol{a}_l}^{k} X^{\boldsymbol{a}_l}) (P_{t_0})\right)
 \, (1, \eta^{\sigma_L (\boldsymbol{a}_l)}, \ldots, \eta^{(q-2) \sigma_L (\boldsymbol{a}_l)}) =
 \left(\mathcal{T}_{\boldsymbol{a}_l} (\xi_{\boldsymbol{a}_l}^{k} X^{\boldsymbol{a}_l}) (P_{t_0})\right) \, \boldsymbol{w}_{l }
$$
for $1\le l\le z$ and   $0 \leq k \leq i_{\boldsymbol{a}_l} -1$, where $\boldsymbol{w}_{l }= (1, \eta^{\sigma_L (\boldsymbol{a}_l)}, \ldots, \eta^{(q-2) \sigma_L (\boldsymbol{a}_l)})$.
Since $\eta$ is a primitive element of $\mathbb{F}_q$, then the set of vectors $\{ \boldsymbol{w}_{1},\boldsymbol{w}_{2},\dots,\boldsymbol{w}_{z} \}$ has rank exactly  $r$.
 Thus, to prove our claim on the dimension of $\mathcal{C}^{J}_\Delta[\mbox{coord}(R_{t_0})]$ it suffices to prove that for any $\boldsymbol{a}\in \{ \boldsymbol{a}_1,\boldsymbol{a}_2,\dots,\boldsymbol{a}_z  \}$, there exists $k$ such that $\mathcal{T}_{\boldsymbol{a}} \left(\xi_{\boldsymbol{a}}^{k} X^{\boldsymbol{a}}\right) (P_{t_0}) \neq 0$. Suppose, on the contrary, that there is   $\boldsymbol{a}$ such that $\mathcal{T}_{\boldsymbol{a}} \left(\xi_{\boldsymbol{a}}^{k} X^{\boldsymbol{a}}\right) (P_{t_0}) = 0$  for all $k$, $0 \leq k \leq i_{\boldsymbol{a}}-1$.
Let $i=i_{\boldsymbol{a}}$. A simple computation shows that we can write
\begin{equation*}
\begin{aligned}
0=\mathcal{T}_{\boldsymbol{a}} \left(X^{\boldsymbol{a}}\right) (P_{t_0}) & = X^{\boldsymbol{a}}(P_{t_0}) (1 + {b_1} + \cdots + {b_{i-1}}),\\
0=\mathcal{T}_{\boldsymbol{a}} \left(\xi_{\boldsymbol{a}} X^{\boldsymbol{a}}\right) (P_{t_0}) & = X^{\boldsymbol{a}}(P_{t_0}) (\xi_{\boldsymbol{a}} +\xi_{\boldsymbol{a}}^q {b_1} + \cdots + \xi_{\boldsymbol{a}}^{(i-1)q} {b_{i-1}}),\\
& \vdots  \\
0=\mathcal{T}_{\boldsymbol{a}} \left(\xi_{\boldsymbol{a}}^{i-1} X^{\boldsymbol{a}}\right) (P_{t_0}) & = X^{\boldsymbol{a}}(P_{t_0}) (\xi_{\boldsymbol{a}}^{i-1} +(\xi_{\boldsymbol{a}}^{i-1})^q {b_1} + \cdots + (\xi_{\boldsymbol{a}}^{i-1})^{(i-1)q}{b_{i-1}})
\end{aligned}
\end{equation*}
for some elements $b_1,b_2,\dots,b_{i-1} \in \mathbb{F}_q$.  Note that $X^{\boldsymbol{a}}(P_{t_0})\neq 0$ since $J=\{1,2, \ldots, m\}$. Thus  the vector $ \left(1, {b_1}, b_2,\ldots,{b_{i-1}}\right) \neq \boldsymbol{0}$ is a solution of a homogeneous square linear system whose matrix is of Vandermonde type, what is not possible.
\end{proof}

Let us recall here the well-known fact that the dual code of a linear MDS code is again an MDS code.

\begin{pro}
\label{esmds}
Let $J=\{ 1,2,\dots,m\}$ and  $\Delta= \cup_{l=1}^{z} \mathfrak{I}_{\boldsymbol{a}_l}$ be a closed set.
If the set $\{ \sigma_L(\boldsymbol{a}_1), \sigma_L(\boldsymbol{a}_2),\dots,$ $\sigma_L(\boldsymbol{a}_z) \}$
contains exactly $r\le q-2$ distinct values and these values are consecutive integers, then for any coordinate $t_0$, the punctured code $\mathcal{C}^{J}_\Delta[\mbox{\rm coord}(R_{t_0})]$ is an MDS code with parameters $[q-1,r,q-r]$, where  $R_{t_0}$ is the orbit of $P_{t_0}\in Z_J$.
\end{pro}
\begin{proof}
For simplicity suppose that  $\sigma_L(\boldsymbol{a}_1), \sigma_L(\boldsymbol{a}_2),\dots,\sigma_L(\boldsymbol{a}_r)$ are the consecutive integers mentioned in the statement. Since $r\le q-2$ all these integers are distinct modulo $q-1$, and hence
the statements about length and dimension follow from Lemma \ref{dimens}. Let us compute the minimum distance of   $\mathcal{C}^{J}_\Delta[\mbox{coord}(R_{t_0})]$. A generator matrix of this code is
\[
\left(
    \begin{array}{cccc}
    \mathcal{T}_{\boldsymbol{a}_1} \left( X^{\boldsymbol{a}_1}\right) (P_{t_0}) & \eta^{\sigma_L (\boldsymbol{a}_1)} \mathcal{T}_{\boldsymbol{a}_1} \left( X^{\boldsymbol{a}_1}\right) (P_{t_0}) & \cdots & \eta^{(q-2)\sigma_L( \boldsymbol{a}_1)} \mathcal{T}_{\boldsymbol{a}_1} \left( X^{\boldsymbol{a}_1}\right) (P_{t_0})  \\
%    \mathcal{T}_{\boldsymbol{a}_2} \left( X^{\boldsymbol{a}_2}\right) (P_{t_0}) & \eta^{\sum_L \boldsymbol{a}_2} \mathcal{T}_{\boldsymbol{a}_2} \left( X^{\boldsymbol{a}_2}\right) (P_{t_0}) & \cdots & \eta^{\sum_L \boldsymbol{a}_2(q -2)} \mathcal{T}_{\boldsymbol{a}_1} \left( X^{\boldsymbol{a}_2}\right) (P_{t_0}) \\
      \vdots & \vdots &  & \vdots \\
     \mathcal{T}_{\boldsymbol{a}_r} \left( X^{\boldsymbol{a}_r}\right) (P_{t_0}) & \eta^{ \sigma_L (\boldsymbol{a}_r)} \mathcal{T}_{\boldsymbol{a}_r} \left( X^{\boldsymbol{a}_r}\right) (P_{t_0}) & \cdots & \eta^{(q-2)\sigma_L (\boldsymbol{a}_r)} \mathcal{T}_{\boldsymbol{a}_r} \left( X^{\boldsymbol{a}_r}\right) (P_{t_0})\\
    \end{array}
  \right).
\]
If $\mathcal{T}_{\boldsymbol{a}_l} \left( X^{\boldsymbol{a}_l}\right) (P_{t_0}) = 0$ we  can remove the corresponding row in this matrix, hence
we can assume $\mathcal{T}_{\boldsymbol{a}_l} \left( X^{\boldsymbol{a}_l}\right) (P_{t_0}) \neq 0$ for all $\boldsymbol{a}_l$.
To study  the independence of columns, it suffices to consider the matrix
\[
\boldsymbol{A}= \left(
    \begin{array}{cccc}
     1& \eta^{\sigma_L (\boldsymbol{a}_1)} & \cdots & \eta^{(q-2)\sigma_L (\boldsymbol{a}_1)} \\
%    1 & \eta^{\sum_L \boldsymbol{a}_2}  & \cdots & \eta^{ \sum_L \boldsymbol{a}_2(q -2)} \\
      \vdots & \vdots &  & \vdots\\
     1& \eta^{\sigma_L (\boldsymbol{a}_r)} & \cdots & \eta^{\sigma_L ((q-2)\boldsymbol{a}_r)} \\
    \end{array}
  \right).
\]
Since $\sigma_L(\boldsymbol{a}_1), \sigma_L(\boldsymbol{a}_2),\dots,\sigma_L(\boldsymbol{a}_r)$ are  consecutive integers,  any submatrix of $\boldsymbol{A}$ obtained by taking $r$ columns of $\boldsymbol{A}$ has rank $r$, see \cite[Lemma 6.6.5]{VL}. Thus the minimum distance of  $\mathcal{C}^{J}_\Delta[\mbox{coord}(R_{t_0})]^\perp$  is $\ge r+1$. So it has parameters  $[q-1, q-1-r, r+1]$ and then it is an MDS code. Therefore $\mathcal{C}^{J}_\Delta[\mbox{coord}(R_{t_0})]$ is also MDS with parameters $[q-1,r,q-r]$.
\end{proof}

As a direct consequence of this proposition, we have the following theorem.

\begin{teo}
\label{teo29}
Let $J=\{ 1,2,\dots,m\}$ and  $\Delta= \cup_{l=1}^{z} \mathfrak{I}_{\boldsymbol{a}_l}$ be a closed set.
If the set $\{ \sigma_L(\boldsymbol{a}_1), \sigma_L(\boldsymbol{a}_2),\dots,$ $\sigma_L(\boldsymbol{a}_z) \}$
contains exactly $r\le q-2$ distinct values and these values are consecutive integers, then for any coordinate $t_0$, the set $\mbox{\rm coord}(R_{t_0})$ is a $(r,q-r)$ recovery set for $t_0$, where   $R_{t_0}$ is the orbit of $P_{t_0}\in Z_J$.  Consequently $\mathcal{C}^{J}_\Delta$
is an LRC code with locality $(r,q-r)$ and  $r_{q-r-t}\le q-1-t$ for   $t=1,2,\dots,q-r-1$.
\end{teo}
\begin{proof}
The first statement follows from Proposition \ref{esmds}. The second one follows from the definition of $r_t$'s and the fact that
puncturing $t<d$ times an $[n,k,d]$ MDS code gives an $[n-t,k,d-t]$  MDS code.
\end{proof}

Let us note that formulas for the length and dimension of the codes $\mathcal{C}^{J}_\Delta$ are given given by Theorem \ref{ddimension}.  For the opposite, we do not have any explicit bound for its minimum distance (apart from the trivial bound $d(\mathcal{C}^{J}_\Delta)\ge d(E^{J}_\Delta)$). The same happens to most codes obtained as subfield-subcodes.
Therefore, we cannot explicitly calculate formulas for the Singleton defect.  In some of the examples we shall show in  Section \ref{sectres}, these distances are calculated by computer search.
Next we shall show that we can give such explicit formulas in the univariate case, $m=1$, when $\mbox{char}(\mathbb{F}_q)\neq 2$ and $Q=q^2$.

So let $m=1$. Given a closed set  $\Delta=\cup_{l=1}^{r} \mathfrak{I}_{\boldsymbol{a}_l} \subset \mathcal{H}_J$, we define its {\em dual} set as  $\Delta^\perp := \mathcal{H}_J \setminus \cup_{l=1}^r \mathfrak{I}_{n_J -\boldsymbol{a}_l}$.
The dual code $(\mathcal{C}^{J}_\Delta)^\perp$ is related to the dual set $\Delta^\perp $ as follows
\begin{equation}
\label{deltadual}
(\mathcal{C}_\Delta^J)^{\perp} = (E_\Delta^J \cap \mathbb{F}_q^{n_J})^{\perp} = E_{\Delta^{\perp}}^J \cap \mathbb{F}_q^{n_J},
\end{equation}
see  \cite[Proposition 2.4]{QINP2}.

\begin{teo}
\label{la211}
Assume $Q=q^2$ with $q$ odd. Let $m=1$, $J=L=\{1\}$ and $N_1=2(q-1)+1$. Take $r\le q-2$ consecutive integers $\boldsymbol{a}_1=0, \boldsymbol{a}_2=1,\dots, \boldsymbol{a}_r=r-1$,  and let $\Delta= \cup_{l=1}^{r} \mathfrak{I}_{\boldsymbol{a}_l}$.
Then the subfield-subcode $\mathcal{C}^{J}_\Delta$ has  length $n_J = 2(q-1)$, dimension $k= 2r - \lceil \frac{r}{2}\rceil$ and minimum distance $d \geq (q-1) - 2 \lfloor \frac{r-2}{2}\rfloor$. It is a sharp LRC code with  locality $(r,q-r)$ and  $(q-r-1)$-th Singleton defect
\[
D_{q-r -1} \le  \left\lceil \frac{r}{2}\right\rceil + 2 \left\lfloor \frac{r-2}{2}\right \rfloor +1 -r.
\]
\end{teo}
\begin{proof}
The length of  $\mathcal{C}^{J}_\Delta$ is $n_J =N_1-1= 2(q-1)$.
The minimal cyclotomic sets in $\mathcal{H}_J$ are  $\mathfrak{I}_{\boldsymbol{a}}=\{\boldsymbol{a}\}$ when $\boldsymbol{a}$ is even and $\mathfrak{I}_{\boldsymbol{a}}=\{\boldsymbol{a}, q+\boldsymbol{a} -1\}$ when $\boldsymbol{a}$ is odd.
Thus $ \# \Delta = 2r - \left\lceil \frac{r}{2}\right\rceil$ and since $\Delta$ is a closed set, $\mathcal{C}^{J}_\Delta$ has dimension $k = \# \Delta = 2r - \left\lceil \frac{r}{2}\right\rceil$ according to Theorem \ref{ddimension}.
Furthermore, a simple computation shows that the dual set   $\Delta^\perp = \mathcal{H}_J \setminus \cup_{l=1}^r \mathfrak{I}_{n_J -\boldsymbol{a}_l}$ is the union of
$(q-2)- 2\lfloor (r-2)/2 \rfloor$ minimal cyclotomic sets with consecutive representatives. Then from equation (\ref{deltadual}) and Proposition \ref{el22a} we get a bound on the minimum distance of $\mathcal{C}^{J}_\Delta$ as follows
$$
d(\mathcal{C}^{J}_\Delta)= d\left(((\mathcal{C}^{J}_\Delta)^{\perp})^{\perp}\right) \ge d\left((E^{J}_{\Delta^{\perp}})^{\perp}\right)
\geq \left(q-1\right) - 2  \left\lfloor \frac{r-2}{2}\right\rfloor.
$$
From Proposition \ref{el22a} the code $\mathcal{C}^{J}_\Delta$ is sharp.
Finally, according to Theorem \ref{teo29},
$\mathcal{C}^{J}_\Delta$ is an LRC code with locality $(r,\delta)=(r, q-r)$. Based on this locality, the $(q-r-1)$-th Singleton defect of $\mathcal{C}^{J}_\Delta$ satisfies
\[
D_{q-r-1} \le 2(q-1) +1 - \left(
(q-1) - 2 \left\lfloor \frac{r-2}{2}\right\rfloor + 2r - \left\lceil \frac{r}{2}\right\rceil + \left\lceil \frac{2r - \left\lceil \frac{r}{2}\right\rceil}{r} -1 \right\rceil (q-r-1)
\right).
\]
Since
\[
\left\lceil \frac{2r - \left\lceil \frac{r}{2}\right\rceil}{r} -1 \right\rceil = 1,
\]
we get the bound on the defect $D_{q-r-1}$.
\end{proof}

\section{Examples}
\label{sectres}

In this section we give examples of parameters of LRC codes $\mathcal{C}^{J}_\Delta$ over $\mathbb{F}_q$ correcting several erasures, which are obtained as subfield-subcodes of $J$-affine variety codes $E^{J}_\Delta$. The theoretical support for these examples is in theorems  \ref{ddimension}, \ref{teo29}
and \ref{la211}. So the parameters we show arise from these theorems
when they provide these data, and are calculated by computer search otherwise (see below for details). In particular, all codes we present have locality  $(r,q-r)$, where   $r\le q-2$ depends on the cyclotomic sets in $\Delta$. In order to evaluate the quality of these codes for local recovery purposes, we use the Singleton bounds (\ref{SingdeltaEq}) and (\ref{Singdt}), and compute the corresponding Singleton defects as stated in Section \ref{secuno}. Let us recall that when
a code of parameters $[n,k,d]$ has locality $(r,\delta)$, then $r_1\le r$ and its defect $D_{\delta-1}$ satisfies
\begin{equation}
\label{estimate}
D_{\delta -1} \le n+1 - \left( d+k+\left( \left\lceil \frac{k}{r}\right\rceil-1\right) (\delta-1) \right).
\end{equation}
In most cases, the codes we present are optimal ($D_{\delta-1}=0$); otherwise they have small defect $D_{\delta-1}$.
Let us remark that, besides the Singleton-like bound, there exist other available bounds on the $(r,\delta)$'s
(eg. Corollary 3 of \cite{ABHM}). To gain clarity, these bounds are not included here.

This section is organized in two subsections. In the first one we show examples with $m=1$; the second one contains examples of bivariate codes, $m=2$,  improving some results obtained in the univariate case. We also include an example showing that this improvement does not always happen.

\subsection{Examples of LRC codes coming from the univariate case}

We take $m=1$, $J=L=\{ 1\}$,  and apply Theorems \ref{teo29} and \ref{la211}.
The following tables contain the relevant data of the obtained codes:   the cardinality $q$ of the ground field over which the code is defined; their parameters $[n,k,d]$; the dual distance $d^{\perp}$; the locality $(r, \delta)$ given by Theorems \ref{teo29} or \ref{la211};  and the estimate on the  $(\delta-1)$-th Singleton defect $D_{\delta-1}$ given by the bound of Equation (\ref{estimate}).
 All these data have been computed from the above theorems, except the minimum distance in Table 1,  which has been obtained by computer search using Magma \cite{Magma}.

Let us note that in all cases,  the localities $(r,\delta)$ shown satisfy that $\delta=q-r$  and $r$ is equal to $d^{\perp}-1$. So from Proposition \ref{d1dual}  it follows that all our codes are sharp and $r=r_1$, the classical locality.

\begin{exa}
{\rm
Let $q=8, Q=64$, $N_1=22$,
$\boldsymbol{a}_l=l-1$, $1\le l\le 6$,
and let $\Delta= \cup_{l=1}^{r} \mathfrak{I}_{\boldsymbol{a}_l}$ for $r=2,3,4,6$. We get codes with the parameters given in Table 1.
\begin{table}[htbp]
\begin{center}
\begin{tabular}{ccccc}
\hline
$q$ & $[n,k,d]$ & $d^{\perp}$  & $(r,\delta)$ &  $D_{\delta -1}$\\
\hline
8 & [21,  3  14] & 3& (2,6) & 0   \\
8 & [21,  5  12] & 4& (3,5) & 1   \\
8 & [21,  6  12] & 5& (4,4) & 1  \\
8 & [21, 10  8] & 7& (6,2) & 3   \\ \hline
\end{tabular}
\end{center}
\label{T1}
\caption{Univariate LRC codes over $\mathbb{F}_8$.}
\end{table}

Look, for example, at the $[21,6,12]$ code of the third row of Table 1. Its has locality $(4,4)$, so $r_3\le 6$ and thus $r_2\le 5, r_1\le 4$. Since $d^{\perp}=5$, from Proposition \ref{dtdual} we get equality in all cases. Then this code has defects $D_1=3, D_2=2$ and $D_3=1$. Note furthermore that the best known  $[21,6]$ code over $\mathbb{F}_8$ has minimum distance $d=13$,  \cite{Mint}.  So no currently  known $[21,6]$ code can be 1-optimal.
}\end{exa}

\begin{exa}{\rm
Let $q= 9$,  $Q=81$,  $N_1=2(q-1)+1=17$,
$\boldsymbol{a}_l=l-1$, $1\le l\le 7$,
and let $\Delta= \cup_{l=1}^{r} \mathfrak{I}_{\boldsymbol{a}_l}$ for $r=2,3,4,5,6,7$.
We get codes with the parameters given in Table \ref{T2}.
\begin{table}[htbp]
\begin{center}
\begin{tabular}{ccccc}
\hline
$q$ & $[n,k,d]$ & $d^{\perp}$ & $(r,\delta)$ &  $D_{\delta -1}$\\
\hline \hline
9 & [16,3,8] & 3& (2,7)   & 0   \\
9 & [16,4,8] & 4 & (3,6) & 0    \\
9 & [16,6,6] & 5 & (4,5) & 1    \\
9 & [16,7,6] & 6 & (5,4) & 1    \\
9 & [16,9,4] & 7 & (6,3) & 2    \\
9 & [16,10,4] & 8 & (7,2) & 2    \\ \hline
\end{tabular}
\end{center}
\caption{Univariate LRC codes over $\mathbb{F}_9$.}
\label{T2}
\end{table}
}\end{exa}

\begin{exa}{\rm
Let $q= 11$, $Q=121$,  $N_1=2(q-1)+1=21$,
$\boldsymbol{a}_l=l-1$, $1\le l\le 7$,
and let $\Delta= \cup_{l=1}^{r} \mathfrak{I}_{\boldsymbol{a}_l}$ for $r=2,3,\dots,7$.
We get codes with the parameters given in Table \ref{T5}.
\begin{table}[htbp]
\begin{center}
\begin{tabular}{ccccc}
\hline
$q$ & $[n,k,d]$ & $d^{\perp}$ &  $(r,\delta)$ &  $D_{\delta -1}$\\
\hline
11 & [20,  3,  10] & 3 & (2,9)  & 0   \\
11 & [20,  5,  10] & 4 & (3,8)	 & 0  \\
11 & [20,  6,  8 ]& 5 & (4,7) & 1    \\
11 & [20,  7,  8] & 6 & (5,6) & 1   \\
11 & [20,  9,  6] & 7 & (6,5) & 2   \\
11 & [20,  10, 6] & 8 & (7,4) & 2   \\ \hline
\end{tabular}
\end{center}
\caption{Univariate LRC codes over $\mathbb{F}_{11}$.}
\label{T5}
\end{table}
}\end{exa}

\begin{exa}{\rm
Let $q= 25$, $Q=625$,  $N_1=49$,
$\boldsymbol{a}_l=l-1$, $1\le l\le 7$,
and let $\Delta= \cup_{l=1}^{r} \mathfrak{I}_{\boldsymbol{a}_l}$ for $r=2,3,\dots,7$.
We get codes with the parameters given in Table \ref{T3}.
\begin{table}[htbp]
\begin{center}
\begin{tabular}{ccccc}
\hline
$q$ & $[n,k,d]$ & $d^{\perp}$  & $(r,\delta)$ &  $D_{\delta -1}$\\
\hline
25 & [48,3,24]& 3& (2,23)   & 0   \\
25 & [48,4,24] & 4& (3,22)	   & 0  \\
25 & [48,6,22] & 5 & (4,21) & 1    \\
25 & [48,7,22] & 6 & (5,20) & 1    \\
25 & [48,9,20] & 7 & (6,19) & 2    \\
25 & [48,10,20] & 8 & (7,18) & 2    \\
\hline
\end{tabular}
\end{center}
\caption{Univariate LRC codes over $\mathbb{F}_{25}$.}
\label{T3}
\end{table}
}\end{exa}

\begin{exa}{\rm
Let $q=27$,   $Q=729$,  $N_1=53$,
$\boldsymbol{a}_l=l-1$, $1\le l\le 7$,
and let $\Delta= \cup_{l=1}^{r} \mathfrak{I}_{\boldsymbol{a}_l}$ for $r=2,3,\dots,7$.
We get codes with the parameters given in Table \ref{T4}.
\begin{table}[htbp]
\begin{center}
\begin{tabular}{ccccc}
\hline
$q$ & $[n,k,d]$ & $d^{\perp}$  &  $(r,\delta)$ &  $D_{\delta -1}$\\
\hline
27 & [52,3,26]& 3& (2,25)  & 0 \\
27 & [52,4,26] & 4& (3,24)	& 0 \\
27 & [52,6,24] & 5 & (4,23) & 1 \\
27 & [52,7,24] & 6 & (5,22) & 1 \\
27 & [52,9,22] & 7 & (6,21) & 2 \\
27 & [52,10,22] & 8 & (7,20) & 2 \\
\hline
\end{tabular}
\end{center}
\caption{Univariate LRC codes over $\mathbb{F}_{27}$.}
\label{T4}
\end{table}
}\end{exa}

\subsection{Examples of LRC codes coming from the bivariate case}
\label{bivariate}

Let us consider now the bivariate case $m=2$, with $J=\{1,2\}$ and $L=\{1\}$. As above we show tables of parameters of codes $\mathcal{C}^{J}_\Delta$ over $\mathbb{F}_q$ for different values of $q$. The minimum distances in these tables  have been computed with Magma \cite{Magma}. The dual distances $d^{\perp}$ have been also computed with Magma;  they are included in the tables only when they provide relevant information about sharpness, what corresponds to the cases $q=8,11,16$. For $q=25,27$, the codes are far from being sharp, and the value of $d^{\perp}$ is omitted.
The examples we are going to present seem to suggest that bivariate codes give better results than the univariate ones. However, this is not always true, as shown in Example \ref{exno}.

\begin{exa}
{\rm
Let $q=8, Q=q^4=4096, N_1= 8$ and $N_2=6$. Table \ref{NT10} contains the parameters of codes $\mathcal{C}^{J}_\Delta$ obtained by using successively the following defining sets $\Delta$: $\mathfrak{I}_{(0,1)}$; $\mathfrak{I}_{(0,1)} \cup \mathfrak{I}_{(1,1)}$; and $\mathfrak{I}_{(0,1)} \cup \mathfrak{I}_{(1,1)} \cup \mathfrak{I}_{(2,1)}$.
\begin{table}[htbp]
\begin{center}
\begin{tabular}{ccccc}
\hline
$q$ & $[n,k,d]$ & $d^{\perp}$ &  $(r,\delta)$ &  $D_{\delta -1}$\\
\hline
8 & [35,4,14] & 2 & (1,7)  & 0   \\
8 & [35,8,12]  & 3 & (2,6)	 & 1  \\
8 & [35,12,10] & 4 & (3,5) & 2    \\
\hline
\end{tabular}
\end{center}
\caption{Bivariate LRC codes over $\mathbb{F}_{8}$.}
\label{NT10}
\end{table}
}\end{exa}

\begin{exa}
{\rm
Let $q=Q=N_1= 11$ and $N_2=3$. Table \ref{T9} contains the parameters of codes $\mathcal{C}^{J}_\Delta$ obtained by using the following defining sets $\Delta$: the first code (first row of Table \ref{T9}) comes from the set
$\mathfrak{I}_{(0,0)} \cup \mathfrak{I}_{(0,1)}  \cup\mathfrak{I}_{(1,0)} \cup\mathfrak{I}_{(2,0)}$;
the defining sets for the remaining codes are obtained by successively adding
the cyclotomic sets
$\mathfrak{I}_{(3,0)}$; $\mathfrak{I}_{(4,0)}$; $\mathfrak{I}_{(5,0)}$;  $\mathfrak{I}_{(1,1)}\cup \mathfrak{I}_{(6,0)}$; and $\mathfrak{I}_{(7,1)}$.
\begin{table}[htbp]
\begin{center}
\begin{tabular}{ccccc}
\hline
$q$ & $[n,k,d]$ & $d^{\perp}$ & $(r,\delta)$ &  $D_{\delta -1}$\\
\hline
11 & [20,4,10]  & 4 & (3,8)  & 0 \\
11 & [20,5,10]  & 4&  (4,7)  & 0 \\
11 & [20,6,10]  & 4 & (5,6) & 0   \\
11 & [20,7,10]  & 4 & (6,5)  & 0   \\
11 & [20,9,8]    & 6 & (7,4)	 & 1  \\
11 & [20,10,8]  & 8 & (8,3)  & 1   \\
\hline
\end{tabular}
\end{center}
\caption{Bivariate LRC codes over $\mathbb{F}_{11}$.}
\label{T9}
\end{table}
}\end{exa}

\begin{exa}
{\rm
Let $q=Q=N_1= 16$ and $N_2=4$. Table \ref{T8} contains the parameters of codes $\mathcal{C}^{J}_\Delta$ obtained by using the following defining sets $\Delta$: the first code comes from the set
$\mathfrak{I}_{(0,0)} \cup \mathfrak{I}_{(0,1)} \cup\mathfrak{I}_{(1,1)} \cup\mathfrak{I}_{(2,0)} \cup\mathfrak{I}_{(3,0)}$.
The defining sets for the remaining ones are obtained by successively adding the cyclotomic sets
$\mathfrak{I}_{(4,0)}$; $\mathfrak{I}_{(5,0)}$; and $\mathfrak{I}_{(4,1)}\cup \mathfrak{I}_{(6,1)}$.
\begin{table}[htbp]
\begin{center}
\begin{tabular}{ccccc}
\hline
$q$ & $[n,k,d]$ & $d^{\perp}$ & $(r,\delta)$ &  $D_{\delta -1}$\\
\hline
16 & [45,5,30]  &3 & (4,12)  & 0 \\
16 & [45,6,30] & 3& (5,11) & 0   \\
16 & [45,7,30] & 3& (6,10)  & 0   \\
16 & [45,9,28]  & 3& (7,9)	 & 1  \\
\hline
\end{tabular}
\end{center}
\caption{Bivariate LRC codes over $\mathbb{F}_{16}$.}
\label{T8}
\end{table}
}\end{exa}

\begin{exa}
{\rm
Let $q=Q=N_1= 25$ and $N_2=3$. Table \ref{T7} contains the parameters of codes $\mathcal{C}^{J}_\Delta$ obtained by using the following defining sets $\Delta$: the first code comes from the set
$\mathfrak{I}_{(0,0)} \cup \mathfrak{I}_{(0,1)} \cup\mathfrak{I}_{(1,0)} \cup\mathfrak{I}_{(1,1)} \cup\mathfrak{I}_{(2,0)}$.
The defining sets for the remaining ones are obtained by successively adding the following cyclotomic sets:
$\mathfrak{I}_{(3,0)}$; $\mathfrak{I}_{(4,0)}$;  $\mathfrak{I}_{(5,0)}$;  $\mathfrak{I}_{(6,0)}$;  $\mathfrak{I}_{(7,0)}$;  $\mathfrak{I}_{(8,0)}$; and  $\mathfrak{I}_{(9,0)}$.
\begin{table}[htbp]
\begin{center}
\begin{tabular}{cccc}
\hline
$q$ & $[n,k,d]$ & $(r,\delta)$ &  $D_{\delta -1}$\\
\hline
25 & [48,5,23]  & (3,22)  & 0 \\
25 & [48,6,23] & (4,21) & 0   \\
25 & [48,7,23] & (5,20)  & 0   \\
25 & [48,8,23]  & (6,19)	 & 0  \\
25 & [48,9,23]  & (7,18) & 0    \\
25 & [48,10,23]  & (8,17) & 0   \\
25 & [48,11,23] & (9,16) & 0   \\
25 & [48,12,23] & (10,15) & 0   \\
\hline
\end{tabular}
\end{center}
\caption{Bivariate LRC codes over $\mathbb{F}_{25}$.}
\label{T7}
\end{table}
}\end{exa}

\begin{exa}
{\rm
Let $q=Q=N_1= 27$ and $N_2=3$. Table \ref{T6} contains the parameters of codes $\mathcal{C}^{J}_\Delta$ obtained by using the following defining sets $\Delta$: the first code  comes from the set
$\mathfrak{I}_{(0,0)} \cup \mathfrak{I}_{(0,1)} \cup\mathfrak{I}_{(1,0)} \cup\mathfrak{I}_{(1,1)} \cup\mathfrak{I}_{(2,0)} \cup\mathfrak{I}_{(3,0)}$.
The defining sets for the remaining ones are obtained by successively adding the cyclotomic sets $\mathfrak{I}_{(4,0)}$; $\mathfrak{I}_{(5,0)}$;  $\mathfrak{I}_{(6,0)}$;  $\mathfrak{I}_{(7,0)}$;  $\mathfrak{I}_{(8,0)}$;  $\mathfrak{I}_{(9,0)}$; and $\mathfrak{I}_{(10,0)} \cup \mathfrak{I}_{(11,0)}$.
\begin{table}[htbp]
\begin{center}
\begin{tabular}{cccc}
\hline
$q$ & $[n,k,d]$ & $(r,\delta)$ &  $D_{\delta -1}$\\
\hline
27 & [52,6,25]  & (4,23)  & 0 \\
27 & [52,7,25] & (5,22)  & 0   \\
27 & [52,8,25]  & (6,21)	 & 0  \\
27 & [52,9,25]  & (7,20) & 0    \\
27 & [52,10,25]  & (8,19) & 0   \\
27 & [52,11,25] & (9,18) & 0   \\
27 & [52,12,25] & (10,17) & 0   \\
27 & [52,13,25] &  (11,16) & 0   \\
27 & [52,14,25] &  (12,15) & 0   \\
\hline
\end{tabular}
\end{center}
\caption{Bivariate LRC codes over $\mathbb{F}_{27}$.}
\label{T6}
\end{table}
}\end{exa}

\begin{exa} \label{exno}
{\rm To conclude this work we show the parameters of some univariate LRC codes  over $\mathbb{F}_{32}$ which seem not to be improved by the bivariate ones. Let $m=1$.  Take $q=32$, $Q=1024, N_1=94$,
$\boldsymbol{a}_l=l-1$ for $1\le l\le 4$,
and let $\Delta= \cup_{l=1}^{r} \mathfrak{I}_{\boldsymbol{a}_l}$ for $r=2,3,4$.
We get univariate codes with the parameters given in Table \ref{T10}.
\begin{table}[htbp]
\begin{center}
\begin{tabular}{ccccc}
\hline
$q$ & $[n,k,d]$ & $d^{\perp}$ &  $(r,\delta)$ &  $D_{\delta -1}$\\
\hline
32 & [93,3,62] & 3 & (2,30) & 0   \\
32 & [93,5,60] & 4 & (3,29) & 1   \\
32 & [93,6,60] & 5 & (4,28) & 1    \\
\hline
\end{tabular}
\end{center}
\caption{Univariate LRC codes over $\mathbb{F}_{32}$ that cannot be improved by the bivariate ones.}
\label{T10}
\end{table}
Bivariate codes  improving these should come from the choice
$N_1=32$,  $N_2=4$, and we are forced to use the same field extension.  An exhaustive computer search  shows that such bivariate codes do not improve the univariate ones.
}\end{exa}

{\bf Acknowledgement.}
We would like to thank the reviewers of this article for their comments and suggestions, which have contributed to significantly improve it. \newline
This work was written during a visit by the third author to the Department of Mathematics at  University Jaume I. This author wishes to express his gratitude to the staff of this university for their hospitality.

\end{document}